\documentclass[conference]{IEEEtran}

\usepackage[cmex10]{amsmath}
\usepackage{amsthm,amsfonts,amssymb,mathtools,mathdots}

\usepackage{enumitem, framed}
\PassOptionsToPackage{svgnames}{xcolor}
\usepackage{tikz}
\usepackage{stmaryrd, tikz-cd, braket}

\newtheorem{thm}{Theorem}
\newtheorem{innercustomthm}{Rule}
\newenvironment{rul}[1]
  {\renewcommand\theinnercustomthm{#1}\innercustomthm}
  {\endinnercustomthm}
\newtheorem{cor}[thm]{Corollary}
\newtheorem{prop}[thm]{Proposition}
\newtheorem{lem}[thm]{Lemma}
\theoremstyle{definition}
\newtheorem{dfn}{Definition}
\theoremstyle{remark}
\newtheorem{remark}{Remark}

\newcommand\restr[2]{{
  \left.\kern-\nulldelimiterspace 
  #1 
  \vphantom{\big|} 
  \right|_{#2}
  }}

\newcommand\ketbra[2]{\ket{#1} \! \! \bra{#2}}

\usetikzlibrary{decorations.markings}
\usetikzlibrary{shapes.geometric}
\usetikzlibrary{patterns}

\pgfdeclarelayer{edgelayer}
\pgfdeclarelayer{nodelayer}
\pgfsetlayers{edgelayer,nodelayer,main}

\tikzstyle{none}=[inner sep=0pt]

\tikzstyle{white}=[circle,fill=White,draw=Black, opacity=0.7, inner sep=0.16em]
\tikzstyle{black}=[circle,fill=Black, inner sep=0.16em]
\tikzstyle{black_fade}=[circle,fill=Black, inner sep=0.16em, opacity=0.25]
\tikzstyle{white_fade}=[circle,fill=White,draw=Black, opacity=0.25, inner sep=0.16em]
\tikzstyle{trian}=[regular polygon,regular polygon sides=3,shape border rotate=0,fill=White, draw=Black, opacity=0.7, inner sep=0.12em]
\tikzstyle{square}=[rectangle,fill=White,draw=Black, opacity=0.7, inner sep=0.24em]

\tikzstyle{simple}=[-,draw=Black, opacity=0.7]
\tikzstyle{simple_fade}=[-,draw=Black, opacity=0.25]
\tikzstyle{simple_b}=[-, shorten >=3pt, draw=Black, opacity=0.7]
\tikzstyle{twohead_b}=[-, shorten >=3pt, shorten <=3pt, draw=Black, opacity=0.7]
\tikzstyle{simple_b_fade}=[-, shorten >=3pt, draw=Black, opacity=0.25]
\tikzstyle{twohead_b_fade}=[-, shorten >=3pt, shorten <=3pt, draw=Black, opacity=0.25]
\tikzstyle{arrow}=[->, shorten >=3pt, shorten <=3pt, dash pattern=on 3pt off 3pt, draw=Black, opacity=0.4]

\tikzset{every picture/.style={scale=0.8}}

\begin{document}

\title{A Diagrammatic Axiomatisation \\
for Qubit Entanglement}

\author{\IEEEauthorblockN{Amar Hadzihasanovic}
\IEEEauthorblockA{University of Oxford, Department of Computer Science\\
Wolfson Building, Parks Road, OX1 3QD Oxford\\
\emph{amarh@cs.ox.ac.uk}}}

\maketitle

\begin{abstract}
Diagrammatic techniques for reasoning about monoidal categories provide an intuitive understanding of the symmetries and connections of interacting computational processes. In the context of categorical quantum mechanics, Coecke and Kissinger suggested that two 3-qubit states, GHZ and W, may be used as the building blocks of a new graphical calculus, aimed at a diagrammatic classification of multipartite qubit entanglement that would highlight the communicational properties of quantum states, and their potential uses in cryptographic schemes.

In this paper, we present a full graphical axiomatisation of the relations between GHZ and W: the ZW calculus. This refines a version of the preexisting ZX calculus, while keeping its most desirable characteristics: undirectedness, a large degree of symmetry, and an algebraic underpinning. We prove that the ZW calculus is complete for the category of free abelian groups on a power of two generators - ``qubits with integer coefficients'' - and provide an explicit normalisation procedure. 
\end{abstract}

\section{Introduction}

After a certain number of quantum systems have interacted with each other, the results of observations on the individual systems may present correlations that cannot possibly be explained by their local features (a ``hidden variable theory''). This phenomenon is called \emph{quantum non-locality}.

Even though it is appealing to see these correlations as a form of ``instantaneous communication'' between systems - whereby it is the actions of one observer that inform the observations of another, however distant they may be - no information can actually be transmitted from one location to the other, in this way. 

In a broader sense, however, communication is about obtaining to share some knowledge - and these \emph{entangled} states can be used as generators of shared information. This is the idea behind entanglement-based key exchange protocols, pioneered by Ekert's E91 \cite{ekert1991quantum}.

Arguably, the kind of information that we could expect to share from a distance can all be digitised. Hence, for purposes of communication, we can restrict our attention to \emph{qubits}: quantum systems that only accept yes/no questions.

When only two users are involved, there is not much else that one can desire, besides perfect correlation. With three or more users, however, a variety of scenarios may arise. Suppose the third user decides to not cooperate: are the other two still to obtain some sharing of information, or should they remain empty-handed?

It turns out that there exists a 3-qubit state, the GHZ state \cite{greenberger1990bell}, for which an uncooperating user results in the other two being disconnected, and another 3-qubit state, the W state, where a communication channel persists between any pair of users. So, one is faced with the following problem:
\begin{itemize}
	\item \emph{find a classification of $n$-qubit entangled states which reflects their different \emph{communicational} properties, and potential uses in quantum cryptography.} 
\end{itemize}

Clearly, we should allow for some local ``pre-processing'' by individual users, prior to measurement: as long as the performed operations are invertible, this can be seen just as a translation of the system into their preferred \emph{format}, not affecting the communication.

If we ask that this pre-processing be a part of the protocol, that is, pre-determined, we obtain the so-called \emph{LOCC classification} of quantum states (\emph{Local Operations and Classical Communication}). If we only want users to perform the ``correct'' pre-processing with non-zero probability, we obtain the \emph{SLOCC classification} of quantum states, which corresponds to allowing arbitrary invertible local operations on the systems \cite{dur2000three}.

For 2 qubits, only two SLOCC classes exist, corresponding to the ``connected'' and ``disconnected'' cases, respectively. For 3 qubits, the GHZ and W states are representatives of the only two classes of connected states. For $n\geq 4$ qubits, however, there are \emph{infinite} classes, as can be shown by a simple count of degrees of freedom \cite{dur2000three}, and only inductive classifications in super-classes, with little insight about operational behaviour, are known \cite{lamata2006inductive, lamata2007inductive}. 

Quantum systems satisfy a property called \emph{map-state duality}: it is possible to turn any input of a quantum operation into an output, and vice versa, so that - for instance - any tripartite state may also be seen as a \emph{binary operation}. In \cite{coecke2010compositional}, Coecke and Kissinger showed that the GHZ and W states correspond, as binary operations, to certain \emph{Frobenius algebras} - in a particular sense, the only two possible kinds of Frobenius algebras on qubits. Moreover, as quantum gates, together with single-qubit states, they are universal for quantum computing, which suggested they could be used as \emph{building blocks} for a compositional classification of multipartite entanglement. 

Coecke and Kissinger formulated their result in the framework of \emph{categorical quantum mechanics} (CQM), initiated by \cite{abramsky2004categorical}. CQM heavily relies on \emph{string diagrams} as a graphical language for monoidal categories \cite{selinger2011survey}: while the latter are a natural home to the formalisation of computational processes and their interactions, the diagrams provide a high-level language for reasoning about them, which bypasses some of the bookkeeping that is associated with algebraic category theory, allowing one to focus on the connections, and the flow of information between such processes. 

In particular, a diagrammatic theory of Frobenius algebras is the basis of the \emph{ZX calculus} \cite{coecke2008interacting}, whose completeness for the important \emph{stabiliser} fragment of quantum mechanics has recently been proven \cite{backens2013zx}. In \cite{coecke2011ghz, kissinger2012pictures}, a graphical axiomatisation of the relations between the GHZ and W algebras was started, with a similar calculus in mind; but this was not brought to completion, and only results about universality and classification were obtained.

In this paper, we present the \emph{ZW calculus}, a diagrammatic axiomatisation of the relations between the GHZ and W algebras, which incorporates a version of the ZX calculus and shares some of its best properties, such as
\begin{itemize}
	\item featuring \emph{undirected} diagrams, that are ``as symmetrical as they look'', and
	\item having a small number of graphical elements and axioms,
	\item described in terms of important algebraic structures and relations.
\end{itemize}

We prove that the ZW calculus is complete for the category $\mathbf{Ab}_{2,\mathrm{free}}$ of free abelian groups on a power of two generators; more informally, it is complete for ``qubits with integer coefficients'', which embed into generic qubits through the inclusion of integers into complex numbers. We achieve this by providing a normal form for string diagrams, and an explicit normalisation procedure. We also derive completeness results for mild extensions of the calculus.

The hope is that, having a complete axiomatisation and a workable calculus at hand, it will be possible to focus on rewrite strategies that are tailored specifically to identifying the SLOCC class of a state, whose communicational properties should be easily read off a properly normalised diagram. These strategies could be then implemented in automated graph-rewriting software, like Quantomatic \cite{kissinger2014quantomatic}. 

\subsubsection*{Background}
While familiarity with graphical languages for monoidal categories would help, the paper only presupposes a basic knowledge of algebra and category theory, including the definition of (symmetric) monoidal category and monoidal functor. 

String diagrams are featured prominently; there are limitations on how well one can portray spatial reasoning on paper, but we tried to give them an appearance of depth, using different visual devices. In particular, we often draw parts of diagrams in a lighter shade, putting them ``in the background'', either to convey that they are not the current focus of a computation, or that their precise structure is not important. Likewise, if we want to suggest that a certain pattern is repeated $n$ times, we only draw the extremities, and one repetition in a lighter shade, followed by the number $n$.

\section{The language}

We describe the ZW calculus in the framework of PROPs \cite{maclane1965categorical}, thinking of its basic diagrams as the \emph{generators}, and its rules as the \emph{relations} that make up the presentation of a PROP. A PROP is a symmetric strict monoidal category that has $\mathbb{N}$ as its set of objects, and a monoidal product given, on objects, by the sum of natural numbers; morphisms $n \to m$ are meant to represent operations with $n$ inputs and $m$ outputs. Diagrammatically, these are depicted as vertices with $n$ incoming wires and $m$ outgoing wires (time flows from bottom to top):
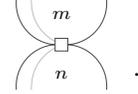
\begin{equation*}
	\input{img/operations.tikz}
\end{equation*}
Composition is depicted as the vertical ``plugging'' of wires, and monoidal product as the horizontal juxtaposition of diagrams. 

A little uncustomarily, we will depict the symmetric braiding as
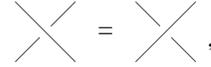
\begin{equation*}
	\input{img/1c_crossing.tikz}
\end{equation*}
in order to leave ``intersecting wires'' available for a different morphism, the \emph{crossing}.

Let $\mathbf{SD}$ be the free \emph{self-dual}, compact closed PROP, that is, a PROP with two generators $\cup: 0 \to 2$, $\cap: 2 \to 0$ satisfying
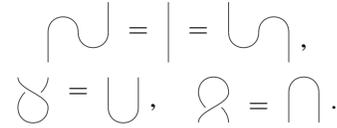
\begin{equation*}
	\input{img/zigzag.tikz}
\end{equation*}
The Kelly-Laplaza coherence theorem for compact closed categories \cite{kelly1980coherence} allows us to be as lax as necessary with the ordering of wires, and the distinction between inputs and outputs, in the diagrammatic calculus. 

If $T$ is a set of operations of a certain arity, let $\mathbf{SD}[T]$ denote the PROP obtained by freely adjoining all morphisms in $T$ to $\mathbf{SD}$. Then, if $R$ is an equivalence relation of morphisms in $\mathbf{SD}[T]$, pairwise of the same type, $\mathbf{SD}[T/R]$ will denote the PROP obtained from $\mathbf{SD}$ by adjoining the generators in $T$ \emph{modulo} the equivalence $R$. We will call pairs in $R$ \emph{rewrite rules}, with graph rewriting in mind (see \cite{kissinger2012pictures} for a review of the subject).

In this formalism, proving the soundness and completeness of the calculus with generators $T$, and rewrite rules $R$, for a monoidal category $\mathbf{C}$, amounts to exhibiting a \emph{monoidal equivalence} between $\mathbf{SD}[T/R]$ and $\mathbf{C}$.

We present two equivalent versions of the ZW calculus. The \emph{condensed} version has the following, infinite set of generators:
\begin{equation*}
	\input{img/3_compact.tikz}
\end{equation*}
The \emph{expanded} version has a finite set of generators $T \subseteq T_c$, containing only binary and ternary black and white vertices.

We interpret these diagrams in $\mathbf{Ab}$, the monoidal category of abelian groups and homomorphisms, with monoidal product given by the tensor product of abelian groups; or rather, in its full subcategory $\mathbf{Ab}_{2,\mathrm{free}}$, generated, under tensoring, by the free abelian group on two generators, $\mathbb{Z}\oplus\mathbb{Z}$.

It is a standard equivalence that abelian groups are the same as $\mathbb{Z}$-modules; the inclusion $\mathbb{Z} \hookrightarrow \mathbb{C}$ induces an inclusion of $\mathbf{Ab}_{2,\mathrm{free}}$ into the category of finite $\mathbb{C}$-modules, that is, complex vector spaces. 

In fact, it is convenient to write the elements of these groups, and the homomorphisms between them, in the \emph{bra-ket} notation commonly used for vectors and linear maps. Hence, letting $\ket{0}$, $\ket{1}$ denote the two generators of $\mathbb{Z}\oplus\mathbb{Z}$, we will write $n\ket{0} + m\ket{1}$, $n, m \in \mathbb{Z}$, for an arbitrary element of $\mathbb{Z}\oplus\mathbb{Z}$; then, we will write $\ket{00} := \ket{0} \otimes \ket{0}$, and $\ketbra{0}{11}$ for the homomorphism $(\mathbb{Z}\oplus\mathbb{Z})^{\otimes 2} \to \mathbb{Z}\oplus\mathbb{Z}$ that sends $\ket{11}$ to $\ket{0}$, and so on.

The semantics of the ZW calculus are defined by a monoidal functor $F: \mathbf{SD}[T_c] \to \mathbf{Ab}_{2,\mathrm{free}}$, fixed by the following interpretation of the generators:
\begin{align*}
	& \input{img/3b1_swap.tikz} \\
	& \begin{tikzpicture}
	\begin{pgfonlayer}{nodelayer}
		\node [style=none] (0) at (-0.5, -0) {};
		\node [style=none] (1) at (0.5, -0) {};
		\node [style=none] (2) at (3, -0.25) {$\mapsto \quad \ket{00} + \ket{11},$};
	\end{pgfonlayer}
	\begin{pgfonlayer}{edgelayer}
		\draw [style=simple, bend right=90, looseness=1.75] (0.center) to (1.center);
	\end{pgfonlayer}
\end{tikzpicture} \\
	& \begin{tikzpicture}
	\begin{pgfonlayer}{nodelayer}
		\node [style=none] (0) at (-0.5, 0) {};
		\node [style=none] (1) at (0.5, 0) {};
		\node [style=none] (2) at (3, 0.25) {$\mapsto \quad \bra{00} + \bra{11},$};
	\end{pgfonlayer}
	\begin{pgfonlayer}{edgelayer}
		\draw [style=simple, bend left=90, looseness=1.75] (0.center) to (1.center);
	\end{pgfonlayer}
\end{tikzpicture} \\
	& \input{img/3b1_black.tikz} \\
	& \input{img/3b1_white.tikz} \\
	& \input{img/3b1_cross.tikz}
\end{align*}
The interpretation of the braiding and of the dualities is self-explanatory. The $n$-ary black vertex corresponds, \emph{modulo} normalisation, to the quantum state $\ket{\mathrm{W}_n}$, the $n$-ary generalisation of the W state. The $n$-ary white vertex, on the other hand, corresponds to the $n$-ary Z spider from the ZX calculus, with a $\pi$ phase \cite{coecke2008interacting}. Save for this phase and normalisation, this is interpreted as the quantum state $\ket{\mathrm{GHZ}_n}$, the $n$-ary generalisation of the GHZ state \cite{greenberger1990bell}.

The \emph{crossing} needs some further explanation. One should keep in mind that this is \emph{not} a braiding in $\mathbf{Ab}_{2,\mathrm{free}}$: such maps have been considered, with applications to supersymmetry \cite{varadarajan2004supersymmetry}, in the theory of \emph{super vector spaces}, or super modules - that is, $\mathbb{Z}_2$-graded modules, with a ``bosonic'' part and a ``fermionic'' part, such that swapping two fermionic states induces a sign change. In our case, $\ket{1}$ would be singled out as the fermionic generator of $\mathbb{Z}\oplus\mathbb{Z}$. However, the categories of super vector spaces and super modules are restricted to so-called \emph{even} maps - that is, maps that preserve the grading: in our case, the ones whose vector expression has an even number of $1$s in each term - for which the crossing is an actual symmetric braiding. 

Moreover, the crossing is not a necessary addition to the graphical language. Crossings are not featured in our normal form for diagrams, and we will provide a systematic procedure for eliminating them. However, their inclusion has some advantages.

First of all, the axioms of the ZW calculus become simpler, and can all be described in terms of well-known algebraic structures and relations, such as Hopf algebras. Only a couple of simple rules needs to be introduced to handle the crossings specifically, plus an elimination rule. 

Furthermore, the binary white vertex has the same interpretation as a self-crossing wire, that is,
	\begin{equation*}
	\input{img/3c_selfcross.tikz}
	\end{equation*}
and can be eliminated in favour of it.

Since the black vertices are interpreted as \emph{odd}, that is, grade-reversing maps - having an odd number of $1$s in each term - this leaves the ternary white vertex as the only vertex, in the expanded calculus, that is interpreted as an ``impure'' map: that is, one which is not even, nor odd. This leaves open the possibility of a \emph{monochromatic} fragment of the ZW calculus, containing only crossings and black vertices, being complete for a subcategory of purely even and purely odd maps.

\section{The rules}

We now present the rule set of the \emph{expanded} ZW calculus. There were some choices to make in its selection and presentation, for which we adopted the following criteria: most subsets of rules should have a short description, linking them to well-known algebraic structures and relations; and the rules of algebraic nature should only contain (weakly) ``planar'' diagrams, that is, diagrams with crossings instead of braidings.

\begin{rul}{0} The black and white vertices are symmetric.
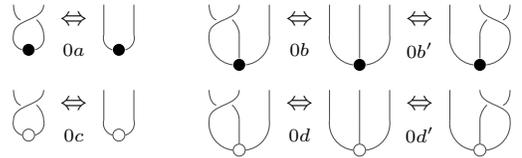
\begin{equation*}
	\input{img/4_rule0.tikz}
\end{equation*}
\begin{equation*}
	\input{img/4_rule0b.tikz}
\end{equation*}
\end{rul}

\begin{remark} This rule allows us to treat the black and white vertices as vertices of an \emph{undirected graph}; in particular, we can turn inputs into outputs, using the dualities, without worrying about which particular wire has been turned around.

For instance, one can speak unambiguously of ``the white vertex with 2 inputs and 1 output'', and depict it as $\;\input{img/white_3.tikz}\;$. We will use Rule 0 implicitly, reshuffling the wires attached to a vertex as needed. 

We will take advantage of this undirectedness throughout, for instance by speaking of \emph{pluggings} of string diagrams, instead of compositions and monoidal products.
\end{remark}

\begin{rul}{1} $\input{img/black_monoid.tikz}$ and $\input{img/white_monoid.tikz}$ are monoids.
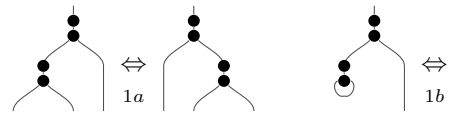
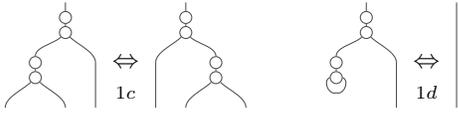
\begin{equation*}
	\input{img/4a_rule1.tikz}
\end{equation*}
\begin{equation*}
	\input{img/4a_rule1b.tikz}
\end{equation*}
\end{rul}

\begin{remark}
Rule 0 implies that the two are actually commutative monoids, which automatically yields the right unitality rules.
\end{remark}

\begin{rul}{2} $\;\begin{tikzpicture}
	\begin{pgfonlayer}{nodelayer}
		\node [style=black] (0) at (-4.5, 0.5) {};
		\node [style=none] (1) at (-4.5, 0.25) {};
		\node [style=none] (2) at (-4.5, 0.75) {};
	\end{pgfonlayer}
	\begin{pgfonlayer}{edgelayer}
		\draw [style=simple] (2.center) to (0);
		\draw [style=simple] (0) to (1.center);
	\end{pgfonlayer}
\end{tikzpicture}\;$ and $\;\begin{tikzpicture}
	\begin{pgfonlayer}{nodelayer}
		\node [style=white] (0) at (-4.5, 0.5) {};
		\node [style=none] (1) at (-4.5, 0.25) {};
		\node [style=none] (2) at (-4.5, 0.75) {};
	\end{pgfonlayer}
	\begin{pgfonlayer}{edgelayer}
		\draw [style=simple] (2.center) to (0);
		\draw [style=simple] (0) to (1.center);
	\end{pgfonlayer}
\end{tikzpicture}\;$ are involutions.
\begin{equation*}
	\input{img/4b_rule2.tikz}
\end{equation*}
\end{rul}

\begin{rul}{3} $\;\begin{tikzpicture}
	\begin{pgfonlayer}{nodelayer}
		\node [style=black] (0) at (-4.5, 0.5) {};
		\node [style=none] (1) at (-4.5, 0.25) {};
		\node [style=none] (2) at (-4.5, 0.75) {};
	\end{pgfonlayer}
	\begin{pgfonlayer}{edgelayer}
		\draw [style=simple] (2.center) to (0);
		\draw [style=simple] (0) to (1.center);
	\end{pgfonlayer}
\end{tikzpicture}\;$ is an automorphism of $\input{img/white_monoid.tikz}$, and $\;
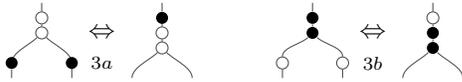
\begin{tikzpicture}
	\begin{pgfonlayer}{nodelayer}
		\node [style=white] (0) at (-4.5, 0.5) {};
		\node [style=none] (1) at (-4.5, 0.25) {};
		\node [style=none] (2) at (-4.5, 0.75) {};
	\end{pgfonlayer}
	\begin{pgfonlayer}{edgelayer}
		\draw [style=simple] (2.center) to (0);
		\draw [style=simple] (0) to (1.center);
	\end{pgfonlayer}
\end{tikzpicture}\;$ of $\input{img/black_monoid.tikz}$.
\begin{equation*}
	\input{img/4b_rule3.tikz}
\end{equation*}
\end{rul}

\begin{remark}
We omitted the rules on $\;\begin{tikzpicture}
	\begin{pgfonlayer}{nodelayer}
		\node [style=black] (0) at (-4.5, 0.5) {};
		\node [style=none] (1) at (-4.5, 0.25) {};
		\node [style=none] (2) at (-4.5, 0.75) {};
	\end{pgfonlayer}
	\begin{pgfonlayer}{edgelayer}
		\draw [style=simple] (2.center) to (0);
		\draw [style=simple] (0) to (1.center);
	\end{pgfonlayer}
\end{tikzpicture}\;$ and $\;\begin{tikzpicture}
	\begin{pgfonlayer}{nodelayer}
		\node [style=white] (0) at (-4.5, 0.5) {};
		\node [style=none] (1) at (-4.5, 0.25) {};
		\node [style=none] (2) at (-4.5, 0.75) {};
	\end{pgfonlayer}
	\begin{pgfonlayer}{edgelayer}
		\draw [style=simple] (2.center) to (0);
		\draw [style=simple] (0) to (1.center);
	\end{pgfonlayer}
\end{tikzpicture}\;$ preserving units, for they are implied by $2a + 3a$, and $2b + 3b$, respectively.
\end{remark}

\begin{rul}{4} $\input{img/white_monoid.tikz}$ and $\input{img/white_comonoid.tikz}$ form a Frobenius algebra.
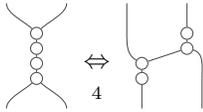
\begin{equation*}
	\input{img/4c_rule4.tikz}
\end{equation*}
\end{rul}

\begin{rul}{5} $\input{img/black_monoid.tikz}$ and $\input{img/black_comonoid.tikz}$ form a Hopf algebra with antipode $\;
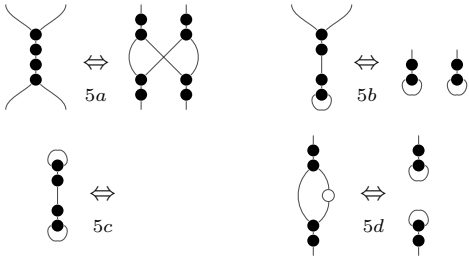
\begin{tikzpicture}
	\begin{pgfonlayer}{nodelayer}
		\node [style=white] (0) at (-4.5, 0.5) {};
		\node [style=none] (1) at (-4.5, 0.25) {};
		\node [style=none] (2) at (-4.5, 0.75) {};
	\end{pgfonlayer}
	\begin{pgfonlayer}{edgelayer}
		\draw [style=simple] (2.center) to (0);
		\draw [style=simple] (0) to (1.center);
	\end{pgfonlayer}
\end{tikzpicture}\;$.
\begin{equation*}
	\input{img/4d_rule5.tikz}
\end{equation*}
\begin{equation*}
	\input{img/4d_rule5b.tikz}
\end{equation*}
\end{rul}

\begin{remark}
I omitted the adjoint (``vertical flip'')  of rule $5b$, which is implied by symmetry.
\end{remark}

\begin{rul}{6} $\input{img/black_monoid.tikz}$ and $\input{img/white_mixed.tikz}$ form a ``Hopf algebra'' with antipode $\;
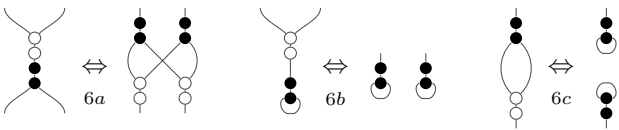
\begin{tikzpicture}
	\begin{pgfonlayer}{nodelayer}
		\node [style=none] (0) at (-4.5, 0.25) {};
		\node [style=none] (1) at (-4.5, 0.75) {};
	\end{pgfonlayer}
	\begin{pgfonlayer}{edgelayer}
		\draw [style=simple] (1.center) to (0.center);
	\end{pgfonlayer}
\end{tikzpicture}\;$.
\begin{equation*}
	\input{img/4e_rule6.tikz}
\end{equation*}
\end{rul}

\begin{remark}
Since $\input{img/white_mixed.tikz}$ is not a comonoid, this is not, properly speaking, a Hopf algebra, but merely a pair satisfying the defining equations of a Hopf algebra. I skipped the two additional equations that coincide with $5c$ and the adjoint of $5b$.
\end{remark}

\begin{rul}{7} $\;\input{img/black_mult.tikz}\;$ is an even map, while $\;\begin{tikzpicture}
	\begin{pgfonlayer}{nodelayer}
		\node [style=black] (0) at (-4.5, 0.5) {};
		\node [style=none] (1) at (-4.5, 0.25) {};
		\node [style=none] (2) at (-4.5, 0.75) {};
	\end{pgfonlayer}
	\begin{pgfonlayer}{edgelayer}
		\draw [style=simple] (2.center) to (0);
		\draw [style=simple] (0) to (1.center);
	\end{pgfonlayer}
\end{tikzpicture}\;$ is odd.
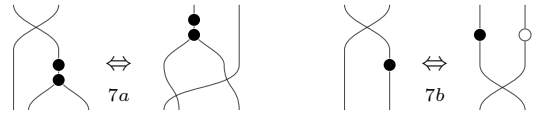
\begin{equation*}
	\input{img/4f_rule7.tikz}
\end{equation*}
\end{rul}

\begin{remark}
The appearance of the white involution - which, as we mentioned, can be replaced with a self-crossing wire - on the other branch of the crossing can be seen as a diagrammatic \emph{definition} of oddness.
\end{remark}

This completes the set of algebraic rules; we single out the last one, which appears to have a purely computational value.
\begin{rul}{X} The elimination rule for crossings.
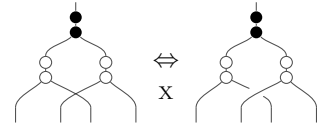
\begin{equation*}
	\input{img/4g_rulex.tikz}
\end{equation*}
\end{rul}

\begin{dfn} The \emph{expanded ZW calculus} is the set $\mathrm{ZW}$ of all rewrite rules contained in Rules 0-7 and X.
\end{dfn}

It can be verified that all rules are sound for our interpretation, that is, the functor $F: \mathbf{SD}[T] \to \mathbf{Ab}_{2,\mathrm{free}}$ commutes through the quotient $\mathbf{SD}[T] \twoheadrightarrow \mathbf{SD}[T/\mathrm{ZW}]$.

While the expanded ZW calculus is complete, it is hardly the most convenient version with which to work, for it does not exploit all the information that can be encoded in the symmetries of string diagrams. The bridge between expanded and condensed diagrams is given by the \emph{spider rules} - actually, rule schemata, for $n, m \in \mathbb{N}$.
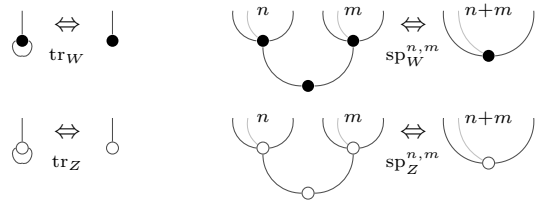
\begin{equation*}
	\input{img/4h_spider_b.tikz}
\end{equation*}
\begin{equation*}
	\input{img/4h_spider_w.tikz}
\end{equation*}
These rules are sound for our interpretation, and, together with rules $2a$ and $2b$, they imply Rule 1, for which they can be substituted:
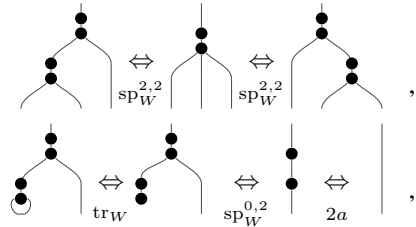
\begin{equation*}
	\input{img/5a_spider_rule1.tikz}
\end{equation*}
and similarly for white vertices. Moreover, Rule 0, together with the spider rules, implies that the $n$-ary vertices are symmetric for all $n \in \mathbb{N}$.

\subsection*{Derived rules}
We now proceed to prove the validity of several useful derived rules.
\begin{lem} \label{prop:phase}
$\;\begin{tikzpicture}
	\begin{pgfonlayer}{nodelayer}
		\node [style=white] (0) at (-4.5, 0.5) {};
		\node [style=none] (1) at (-4.5, 0.25) {};
		\node [style=none] (2) at (-4.5, 0.75) {};
	\end{pgfonlayer}
	\begin{pgfonlayer}{edgelayer}
		\draw [style=simple] (2.center) to (0);
		\draw [style=simple] (0) to (1.center);
	\end{pgfonlayer}
\end{tikzpicture}\;$ commutes with $\;\input{img/white_3.tikz}\;$, that is,
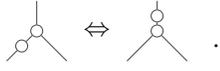
\begin{equation*}
	\input{img/5b_phase.tikz}
\end{equation*}
\end{lem}
\begin{remark}
In the terminology of \cite{coecke2011phase}, $\;\begin{tikzpicture}
	\begin{pgfonlayer}{nodelayer}
		\node [style=white] (0) at (-4.5, 0.5) {};
		\node [style=none] (1) at (-4.5, 0.25) {};
		\node [style=none] (2) at (-4.5, 0.75) {};
	\end{pgfonlayer}
	\begin{pgfonlayer}{edgelayer}
		\draw [style=simple] (2.center) to (0);
		\draw [style=simple] (0) to (1.center);
	\end{pgfonlayer}
\end{tikzpicture}\;$ is a \emph{phase} for $\;\input{img/white_3.tikz}\;$. 
\end{remark}
\begin{proof}
First of all,
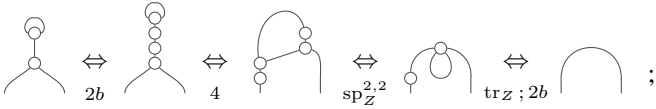
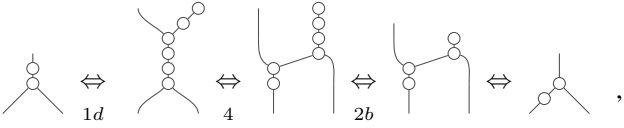
\begin{equation*}
	\input{img/5b_phase_proof.tikz}
\end{equation*}
then,
\begin{equation*}
	\input{img/5b_phase_proof2.tikz}
\end{equation*}
where the last step utilises the previous derivation.
\end{proof}

This derived rule, together with Rule 1, implies Rule 4, and can be used to replace it.

\begin{prop}[Generalised phase rule]
$\;\begin{tikzpicture}
	\begin{pgfonlayer}{nodelayer}
		\node [style=white] (0) at (-4.5, 0.5) {};
		\node [style=none] (1) at (-4.5, 0.25) {};
		\node [style=none] (2) at (-4.5, 0.75) {};
	\end{pgfonlayer}
	\begin{pgfonlayer}{edgelayer}
		\draw [style=simple] (2.center) to (0);
		\draw [style=simple] (0) to (1.center);
	\end{pgfonlayer}
\end{tikzpicture}\;$ commutes with all white vertices: for all $n \in \mathbb{N}$,
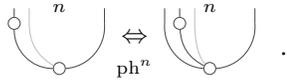
\begin{equation*}
	\input{img/5c_phase_comm.tikz}
\end{equation*}
\end{prop}
\begin{proof} For $n = 0, 1, 2$ there is nothing to prove. For $n > 2$, the claim follows from Lemma \ref{prop:phase}, by
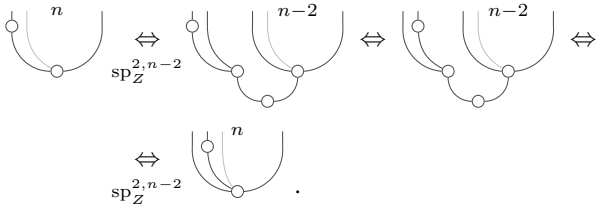
\begin{equation*}
	\input{img/5c_phase_comm2.tikz} \qedhere
\end{equation*}
\end{proof}

The previous is the first of a series of inductive generalisations of the basic rules, with proofs all very similar to each other: we start from the case of ternary vertices, and use the spider rule for the inductive step. We will omit their details.

\begin{prop}[Generalised automorphism rules]
The following are derived rewrite rules, for all $n \in \mathbb{N}$:
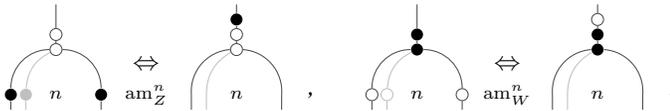
\begin{equation*}
	\input{img/5d_automorph.tikz}
\end{equation*}
\end{prop}
\begin{proof} The cases $n = 0,1,2$ are given by Rules 2 and 3. For $n > 2$, proceed by induction.
\end{proof}

\begin{prop}[Generalised bialgebra rule, I] \label{prop:bialgebra}
The following is a derived rewrite rule, for all $n, m \in \mathbb{N}$:
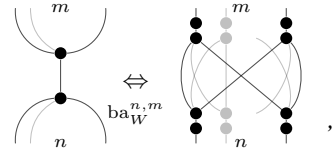
\begin{equation*}
	\input{img/5e_bialgebra.tikz}
\end{equation*}
where, in the RHS, there is a single wire connecting each top vertex to each bottom vertex.
\end{prop}
\begin{proof}
Combined with rule $2a$, the case $n = m = 0$ is rule $5c$; the case $n = 1$ or $m = 1$ is trivial; $n = 0$ and $m > 1$, or vice versa, is an easy inductive generalisation of rule $5b$; and $n = m = 2$ is rule $5a$. From here on, proceed by double induction on $n$ and $m$. 
\end{proof}

\begin{prop}[Generalised loop rule, I]
The following is a derived rewrite rule, for all $n, m \in \mathbb{N}$, $n \geq m$:
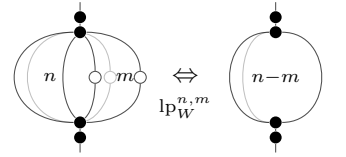
\begin{equation*}
	\input{img/5i_loop.tikz}
\end{equation*}
\end{prop}
\begin{proof}
For $m = 0$, there is nothing to prove. For $m>0$, observe that
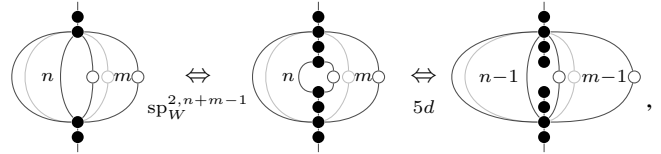
\begin{equation*}
	\input{img/5i_loop_proof.tikz}
\end{equation*}
and use the inductive hypothesis on $m-1$.
\end{proof}

\begin{remark} \label{remark:looprule}
The case $n < m$ can be handled as follows:
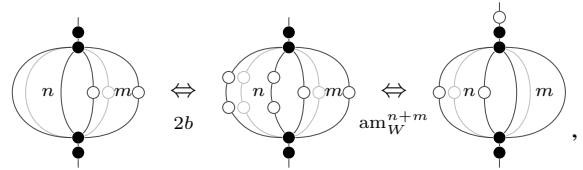
\begin{equation*}
	\input{img/5i_loop_proof2.tikz}
\end{equation*}
and then apply the previous Proposition, recalling the all internal wires can be reordered by symmetry.
\end{remark}

\begin{lem} \label{lem:ba_braiding}
The following is a derived rewrite rule:
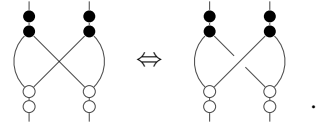
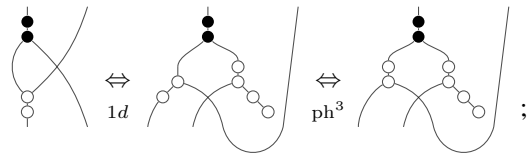
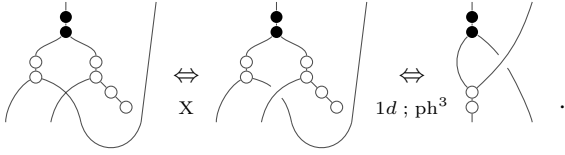
\begin{equation*}
	\input{img/5f_ba_braiding.tikz}
\end{equation*}
\end{lem}
\begin{proof}
We have
\begin{equation*}
	\input{img/5f_ba_braiding_proof.tikz}
\end{equation*}
then
\begin{equation*}
	\input{img/5f_ba_braiding_proof2.tikz}
\end{equation*}
The claim immediately follows.
\end{proof}

\begin{prop}[Generalised bialgebra rule, II] \label{prop:bialgebra2}
The following is a derived rewrite rule, for all $n \in \mathbb{N}$, $m > 0$:
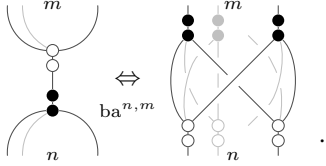
\begin{equation*}
	\input{img/5g_bialgebra2.tikz}
\end{equation*}
\end{prop}
\begin{proof}
The proof is basically the same as that of Proposition \ref{prop:bialgebra}, where we omit the cases with $m = 0$, use rule $6b$ instead of $5b$, and rule $6a$, with a braiding replacing the crossing as by Lemma \ref{lem:ba_braiding}, instead of rule $5a$.
\end{proof}

\begin{remark}
In fact, Proposition \ref{prop:bialgebra2} also holds for $n = m = 0$:
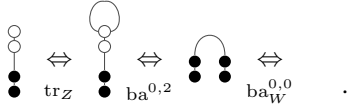
\begin{equation*}
	\input{img/5g_bialgebra2_add.tikz}
\end{equation*}
\end{remark}

\begin{prop}[Generalised loop rule, II]
The following is a derived rewrite rule, for all $n \geq 2$:
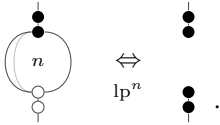
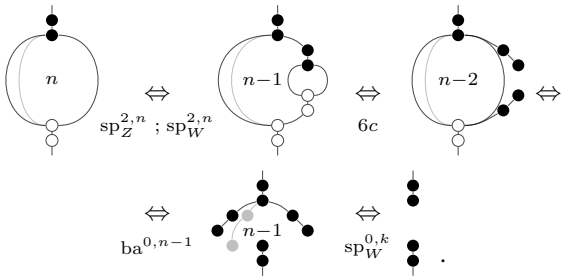
\begin{equation*}
	\input{img/5j_loop2.tikz}
\end{equation*}
\end{prop}
\begin{proof}
Follows from
\begin{equation*}
	\input{img/5j_loop2_proof.tikz} \qedhere
\end{equation*}
\end{proof}

\begin{dfn} The \emph{condensed ZW calculus} is the set $\mathrm{ZW}_c$ consisting of
\begin{enumerate}
	\item the rewrite rules contained in Rules 0, 2, 7 and X, plus $\mathrm{tr}_W$ and $\mathrm{tr}_Z$;
	\item for all $n, m \in \mathbb{N}$, the rules $\mathrm{sp}_W^{n,m}$, $\mathrm{sp}_Z^{n,m}$, $\mathrm{ph}^n$, $\mathrm{am}_W^n$, $\mathrm{am}_Z^n$, $\mathrm{ba}_W^{n,m}$, $\mathrm{lp}_W^{n+m, m}$, $\mathrm{ba}^{n,m+1}$, and $\mathrm{lp}^{n+2}$.
\end{enumerate}
\end{dfn}
We write $F_\mathrm{ZW}: \mathbf{SD}[T_c/\mathrm{ZW}_c] \to \mathbf{Ab}_{2,\mathrm{free}}$ for the functor induced from $F$ by soundness of the rewrite rules.

Even though the condensed ZW calculus has, technically, infinite rewrite rules, which may seem to be a disadvantage, all of its rule schemata are suitable for an implementation using $!$-\emph{graphs} \cite{merry2014reasoning} in Quantomatic. This leads, after all, to a smaller ruleset, as well as shorter derivations.

By the proofs contained in this section, the condensed ZW calculus is equivalent to the expanded ZW calculus, \emph{modulo} the spider rules; that is, $\mathbf{SD}[T/\mathrm{ZW}]$ and $\mathbf{SD}[T_c/\mathrm{ZW}_c]$ are monoidally equivalent PROPs. In the next section, we will prove the completeness of the latter for $\mathbf{Ab}_{2,\mathrm{free}}$, obtaining, at the same time, that of the former.

\section{Completeness}

Any element $\psi$ of $(\mathbb{Z}\oplus\mathbb{Z})^{\otimes n}$ can be uniquely written as a sum
\begin{equation} \label{eq:sum}
	\sum_{i=1}^q \, (-1)^{p_i}\,m_i \, \ket{b_{i,1}\ldots b_{i,n}}\;,
\end{equation}
for some $q \leq 2^n$, $m_i > 0$, and $p_i, b_{i,j} \in \{0,1\}$, $i = 1,\ldots,q$, $j = 1,\ldots,n$, such that no pair of sequences $b_{i,1}\ldots b_{i,n}$ is equal.

We define $N(\psi)$ to be the string diagram
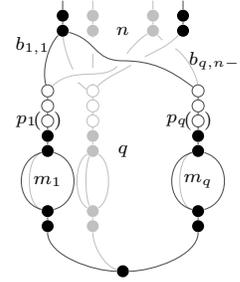
\begin{equation*}
	\input{img/6a_normal.tikz}
\end{equation*}
where
\begin{itemize}
\item the ``sign changer'' vertex $\;\begin{tikzpicture}
	\begin{pgfonlayer}{nodelayer}
		\node [style=white] (0) at (-4.5, 0.5) {};
		\node [style=none] (1) at (-4.5, 0.25) {};
		\node [style=none] (2) at (-4.5, 0.75) {};
	\end{pgfonlayer}
	\begin{pgfonlayer}{edgelayer}
		\draw [style=simple] (2.center) to (0);
		\draw [style=simple] (0) to (1.center);
	\end{pgfonlayer}
\end{tikzpicture}\;$ marked with $p_i$ is only there if $p_i = 1$, and
\item the wire marked with $b_{i,j}$, connecting the $i$th white vertex to the $j$th black vertex, is only there if $b_{i,j} = 1$.
\end{itemize}
By symmetry, the ordering of the internal wires is irrelevant, although it is possible to fix an arbitrary criterion, if needed for uniqueness.

\begin{remark}
	The diagram could be additionally simplified by using the spider rule for black vertices, and rule $2b$ to eliminate some binary white vertices. However, we priviliged this form, for it exposes all the individual computational components.
\end{remark}

All homomorphisms $f: (\mathbb{Z}\oplus\mathbb{Z})^{\otimes n_1} \to (\mathbb{Z}\oplus\mathbb{Z})^{\otimes n_2}$ are the partial transpose of some state $\psi_f$ of $(\mathbb{Z}\oplus\mathbb{Z})^{\otimes (n_1 + n_2)}$, so we can define $N(f)$ to be $N(\psi_f)$ with some of the outputs turned into inputs, using the dualities.

We say that a string diagram $\mathcal{G}$ is in \emph{normal form} if there exists a morphism $f$ of $\mathbf{Ab}_{2,\mathrm{free}}$ such that $\mathcal{G} = N(f)$. 

\begin{remark}
Speaking of a normal form is a slight abuse of terminology, since the term is usually associated to terminating, confluent rewrite systems. However, as long as a directed, confluent version of the ZW calculus has not been developed, it should be acceptable.
\end{remark}

\begin{remark}
An embryo of this normal form appeared in \cite{bruni2006basic}, where an axiomatisation of a subcategory of $\mathbf{FRel}$, the category of finite sets and relations - as modules over the semiring of Booleans - was proposed, using the analogues of the GHZ and W monoids. 

This axiomatisation was complete for the theory considered there, but had a large number of convoluted axioms, including a complicated axiom schema with one rule for all $n \in \mathbb{N}$. However, it stirred further work on algebras of connectors for the study of concurrent systems \cite{bruni2013connector}, which ended up crossing paths with research on the ZX calculus \cite{bonchi2014interacting}.
\end{remark}

We claim that $FN(f) = f$; it suffices to check this for states $\psi \in (\mathbb{Z}\oplus\mathbb{Z})^{\otimes n}$. In fact, we will always consider string diagrams corresponding to states; dualities take care of the general case.

\begin{enumerate}
\item First of all, $\input{img/6b_black.tikz}$ gives a state $\ket{\mathrm{W}_q} = \ket{10\ldots0} + \ldots + \ket{0\ldots01}$. The $i$th individual summand, $\ket{0\ldots010\ldots0}$, has a single $1$ in the $i$th position.

\item Then, for $i = 1,\ldots,q$, 
\begin{equation*}
\input{img/6b_loop.tikz}
\end{equation*}
The $i$th summand is transformed into $(-1)^{p_i}\, m_i \,\ket{0\ldots010\ldots0}$.

\item Finally, the
\begin{equation*}
\input{img/6b_white.tikz}
\end{equation*}
copy both $0$s and $1$s. The $0$s of $\ket{0\ldots010\ldots0}$ get ``absorbed'' by the black vertices:
\begin{equation*}
\input{img/6b_spider.tikz}
\end{equation*}
leaving only a diagram of the form
\begin{equation*}
\input{img/6b_reassign.tikz}
\end{equation*}
Overall, the $i$th summand is transformed into $(-1)^{p_i}\,m_i \, \ket{b_{i,1}\ldots b_{i,n}}$, and $\ket{\mathrm{W}_q}$ into $\psi$.
\end{enumerate}

This proves that our interpretation $F$ is a \emph{full} functor over $\mathbf{Ab}_{2,\mathrm{free}}$. Completeness of the ZW calculus for $\mathbf{Ab}_{2,\mathrm{free}}$ will ensue from the following two facts:
\begin{enumerate}[label=(\alph*)]
\item $N$ is a monoidal functor $\mathbf{Ab}_{2,\mathrm{free}} \to \mathbf{SD}[T_c / \mathrm{ZW}_c]$;
\item $N$ is a left inverse for $F_\mathrm{ZW}$.
\end{enumerate}

\begin{lem}[Delooping] A string diagram in normal form can be rewritten in a loop-free form, that is, \label{lem:delooping}
\begin{equation} \label{eq:delooped}
\input{img/6c_delooped.tikz}
\end{equation}
\end{lem}
\begin{remark} This operation corresponds, basically, to writing $m_i$ as the sum $\overbrace{1 + \ldots + 1}^{m_i}$. \end{remark}
\begin{proof}
Follows from 
\begin{equation*}
\input{img/6c_delooped_proof.tikz}
\end{equation*}
performed on all loops, with a final application of the spider rule to merge all the black vertices on the bottom.
\end{proof}

\begin{remark} \label{remark:prenormal} Conversely, we can rewrite in normal form any diagram that is in a form like (\ref{eq:delooped}), and may additionally
\begin{itemize}
	\item have more than one wire connecting a pair of a black and a white vertex: these can be eliminated with the rules $\mathrm{lp}^n$;
	\item have two white vertices connected to the same outputs, one with a sign changer, the other without it.
\end{itemize}
The latter, intuitively, correspond to a term $1 - 1$ in the summation, and should cancel out. By retracing the proof of Lemma \ref{lem:delooping}, we see that these pairs end up being rewritten as a loop
\begin{equation*}
	\input{img/6c_delooped_add.tikz}
\end{equation*}
to which the rule $\mathrm{lp}_W^{n,m}$ can be applied, either directly, or through the steps of Remark \ref{remark:looprule}. 
\end{remark}

We say that such diagrams are in \emph{pre-normal form}. In most of the following proofs, we will deloop diagrams in normal form, and perform certain operations that will, in general, only yield a diagram in pre-normal form; that this is sufficient follows from the considerations of Remark \ref{remark:prenormal}.

\begin{lem}[Negation]\label{lem:negation} The plugging of $\;\begin{tikzpicture}
	\begin{pgfonlayer}{nodelayer}
		\node [style=black] (0) at (-4.5, 0.5) {};
		\node [style=none] (1) at (-4.5, 0.25) {};
		\node [style=none] (2) at (-4.5, 0.75) {};
	\end{pgfonlayer}
	\begin{pgfonlayer}{edgelayer}
		\draw [style=simple] (2.center) to (0);
		\draw [style=simple] (0) to (1.center);
	\end{pgfonlayer}
\end{tikzpicture}\;$ into one end of a diagram in (pre-)normal form can be rewritten in normal from, and has the effect of ``negating'' its connections to the white vertices; that is,
\begin{equation*}
	\input{img/6d_negation.tikz}
\end{equation*}
\end{lem}
\begin{proof}
Suppose first that $n>0$. By using the spider rules and the phase rules, we can ``detach'' the part of the diagram containing the connections of the vertex that is involved:
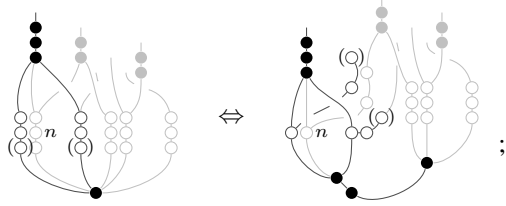
\begin{equation*}
	\input{img/6d_negation_proof.tikz}
\end{equation*}
applying the bialgebra rule $\mathrm{ba}^{2,n}$, this is rewritten as
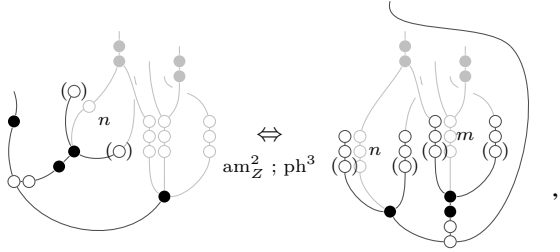
\begin{equation*}
	\input{img/6d_negation_proof2.tikz}
\end{equation*}
where we used the automorphism rule to push $\;\begin{tikzpicture}
	\begin{pgfonlayer}{nodelayer}
		\node [style=black] (0) at (-4.5, 0.5) {};
		\node [style=none] (1) at (-4.5, 0.25) {};
		\node [style=none] (2) at (-4.5, 0.75) {};
	\end{pgfonlayer}
	\begin{pgfonlayer}{edgelayer}
		\draw [style=simple] (2.center) to (0);
		\draw [style=simple] (0) to (1.center);
	\end{pgfonlayer}
\end{tikzpicture}\;$ through, and moved vertices around a bit to make the next step clearer. In the case $n = 0$, we can directly skip to this point:
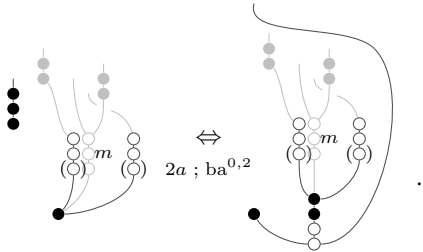
\begin{equation*}
	\input{img/6d_negation_proof2b.tikz}
\end{equation*}

Using the bialgebra rule $\mathrm{ba}^{m,2}$, and rule $2b$ to eliminate some binary white vertices, we rewrite this as
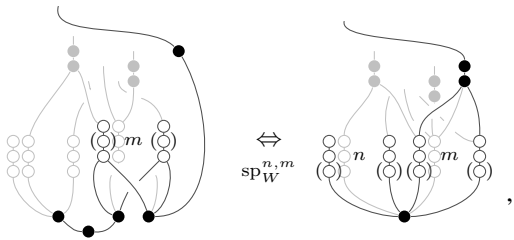
\begin{equation*}
	\input{img/6d_negation_proof3.tikz}
\end{equation*}
which completes the proof.
\end{proof}

In particular, through negation in the sense of this lemma, a disconnected black vertex can be connected to \emph{all} the topmost white vertices.

\begin{lem}[Trace] Let $\mathcal{G}$ be a string diagram in (pre-)normal form. The plugging of two open ends of $\mathcal{G}$ into each other, 
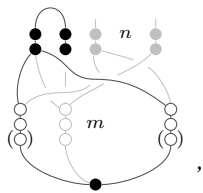
\label{lem:selfplug}
\begin{equation*}
	\input{img/6e_selfplug.tikz}
\end{equation*}
can be rewritten in normal form.
\end{lem}
\begin{proof}
We apply negation repeatedly; since this only affects the connections of the two ends involved, we can avoid drawing the rest of the diagram. We distinguish four groups of white vertices, based on their being connected to both ends, only one end, or no ends:
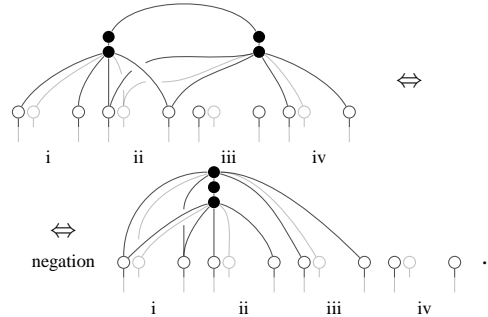
\begin{equation*}
	\input{img/6e_selfplug_proof.tikz}
\end{equation*}
Then, using the spider rule to merge the black vertices, and the $\mathrm{lp}^2$ rule to eliminate group i, we rewrite this as
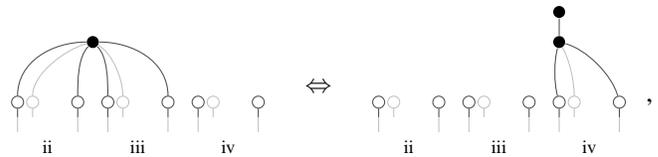
\begin{equation*}
	\input{img/6e_selfplug_proof2.tikz}
\end{equation*}
where we used negation again. Finally, focusing on group iv, 
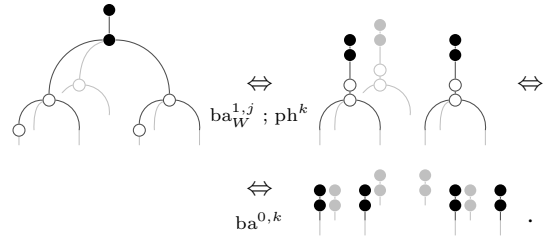
\begin{equation*}
	\input{img/6e_selfplug_proof3.tikz}
\end{equation*}
\emph{Modulo} the automorphism rule, wires on the bottom all lead to black vertices, so we can apply the spider rule, obtaining a diagram in (pre-)normal form.
\end{proof}

The nullary black vertex is interpreted as $0$; the next lemma proves that it acts this way.

\begin{lem}[Absorption] For all diagrams in (pre-)normal form, the following is a valid rewrite rule:
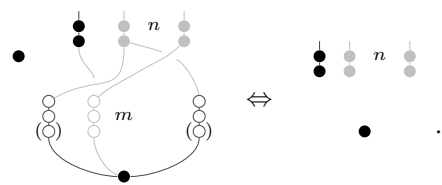
\begin{equation*}
	\input{img/6f_absorption.tikz}
\end{equation*}
\end{lem}
\begin{proof}
Using negation, we obtain
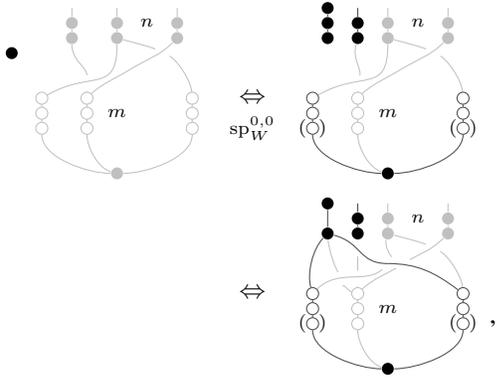
\begin{equation*}
	\input{img/6f_absorption_proof.tikz}
\end{equation*}
where the new vertex is connected to all the topmost white vertices. From here, we can proceed as in the last part of Lemma \ref{lem:selfplug}.
\end{proof}

With the negation, trace and absorption lemmata on hand, we are able to give the central proof of our completeness theorem.

\begin{thm} Let $\mathcal{G}$, $\mathcal{H}$ be two string diagrams in (pre-)normal form. Then the plugging of $\mathcal{G}$ and $\mathcal{H}$ along any number of wires can be rewritten in normal form. \label{thm:plugging}
\end{thm}
\begin{proof}
As usual, deloop $\mathcal{G}$ and $\mathcal{H}$ if they are not already loop-free. Suppose that \emph{no} end of $\mathcal{G}$ is plugged into one of $\mathcal{H}$, that is, the two diagrams are simply juxtaposed. Then, we can rewrite the result as
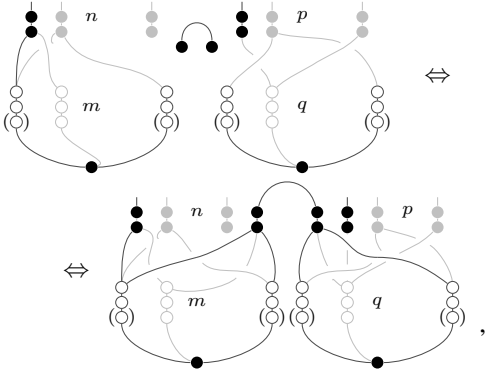
\begin{equation*}
	\input{img/6g_plugging_0.tikz}
\end{equation*}
where we introduced a $\;\begin{tikzpicture}
	\begin{pgfonlayer}{nodelayer}
		\node [style=black] (0) at (0, 0.25) {};
		\node [style=black] (1) at (0, 0.5) {};
	\end{pgfonlayer}
	\begin{pgfonlayer}{edgelayer}
		\draw [style=simple] (1) to (0);
	\end{pgfonlayer}
\end{tikzpicture}\;$ in the picture, using $\mathrm{ba}_W^{0,0}$, and applied negation twice. The diagram so obtained is a plugging along a single end.

On the other hand, if more ends of $\mathcal{G}$ are plugged into ends of $\mathcal{H}$, we can factor the plugging as a single-end plugging, followed by a sequence of traces, as in Lemma \ref{lem:selfplug}. Therefore, it suffices to consider the case where one end of $\mathcal{G}$ is plugged into an end of $\mathcal{H}$.

If one of the open ends that are being plugged is disconnected from white vertices, or both of them are, we can apply $2a$, and then negation to both of them. The only cases when this still leaves one end disconnected are
\begin{enumerate}[label=(\alph*)]
	\item when one of the diagrams is of the form
	\begin{equation*}
	\input{img/6g_zero_diagram.tikz}
	\end{equation*}
	and we can use the nullary black vertex $\;\begin{tikzpicture}
	\begin{pgfonlayer}{nodelayer}
		\node [style=black] (0) at (0, 0.25) {};
	\end{pgfonlayer}
\end{tikzpicture}\;$ to absorb the other diagram;
	
	\item when one end was connected to all the topmost white vertices, and the other to none. In this case, by only negating the first, we can obtain $\;\begin{tikzpicture}
	\begin{pgfonlayer}{nodelayer}
		\node [style=black] (0) at (0, 0.25) {};
	\end{pgfonlayer}
\end{tikzpicture}\;$, and apply the absorption lemma again.
\end{enumerate}
Therefore, we can assume that both ends are connected to at least one white vertex of their respective diagrams.

Focusing on one side of the plugging - say, $\mathcal{G}$ - and the connected white vertices, we have a subdiagram of the form
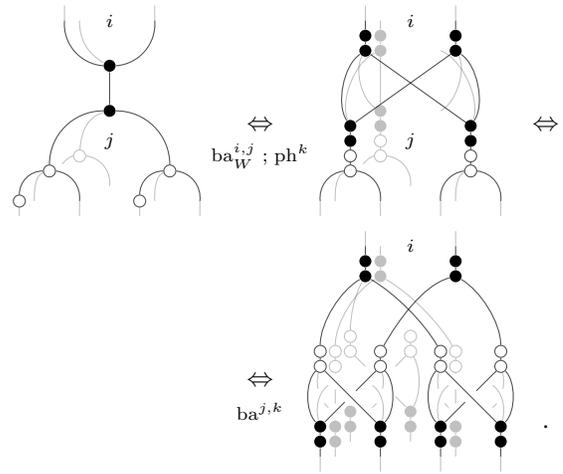
\begin{equation*}
	\input{img/6g_plugging_proof.tikz}
\end{equation*}
\emph{Modulo} the automorphism rule, the wires on the bottom all lead to black vertices, which we can merge with the spider rules. 

In particular, one wire for each of the initial white vertices leads to the ``bottom'' black vertex of $\mathcal{G}$; hence, there is one black vertex to which all of the newly created white vertices are connected. This allows us to use rule X on any pair of white vertices, and turn all the crossings of the diagram into braidings.

Proceeding symmetrically on the side of $\mathcal{H}$, we can push all the white vertices in the middle, at which point, applying the spider rules and $2b$ as much as needed, we obtain a string diagram in pre-normal form. This completes the proof.
\end{proof}

\begin{cor}
$N$ is a monoidal functor $\mathbf{Ab}_{2,\mathrm{free}} \to \mathbf{SD}[T_c / \mathrm{ZW}_c]$.
\end{cor}
\begin{proof}
Both composition and monoidal product in $\mathbf{SD}[T_c / \mathrm{ZW}_c]$ correspond to certain pluggings (possibly along zero wires) of string diagrams. Moreover, by uniqueness of the normal form and soundness of the rewrite rules, if $N(g)\circ N(f)$ is rewritten into $N(h)$ for some homomorphism $h$, then necessarily $h = g\circ f$; similarly for $N(g) \otimes N(f)$.

Since, by Theorem \ref{thm:plugging}, such a rewrite is always possible, it follows that $N$ is a monoidal functor.
\end{proof}

\begin{thm} Every string diagram can be rewritten in normal form. \label{thm:completeness}
\end{thm}
\begin{proof}
First of all, using the spider rules, we can rewrite every string diagram into a diagram with only ternary and binary black and white vertices. Such a diagram is equal to a plugging of generators in $T$, dualities, and braidings; by Theorem \ref{thm:plugging}, it suffices to prove that these can be rewritten in normal form.

The black vertices are trivial:
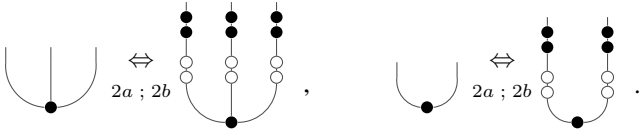
\begin{equation*}
	\input{img/6h_rewrite_black.tikz}
\end{equation*}

For the ternary white vertex:
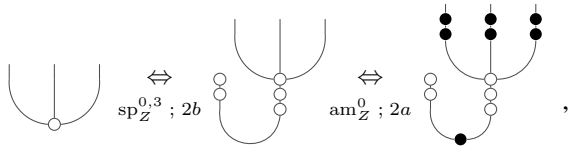
\begin{equation*}
	\input{img/6h_rewrite_white.tikz}
\end{equation*}
and similarly for the binary one. For the dualities:
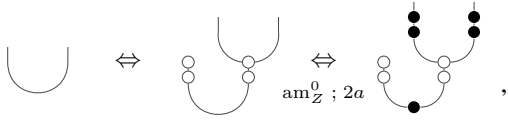
\begin{equation*}
	\input{img/6h_rewrite_wire.tikz}
\end{equation*}
where the first rewrite was derived in the proof of Lemma \ref{prop:phase}.

Finally, for the crossings, observe first that
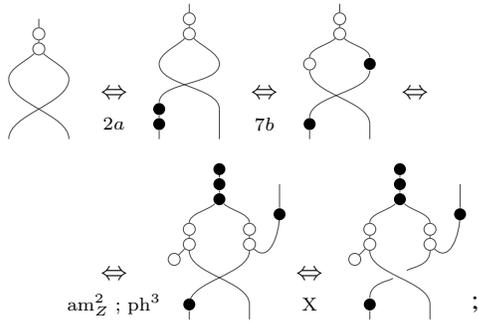
\begin{equation*}
	\input{img/6h_rewrite_crossing1.tikz}
\end{equation*}
from which, retracing some steps, we obtain
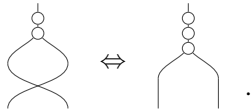
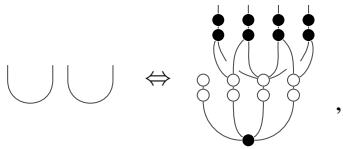
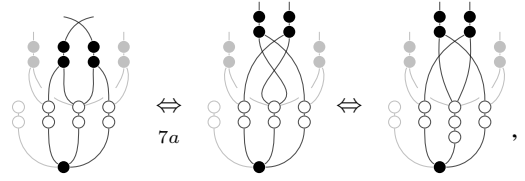
\begin{equation*}
	\input{img/6h_rewrite_crossing2.tikz}
\end{equation*}
Then, having rewritten
\begin{equation*}
	\input{img/6h_rewrite_crossing3.tikz}
\end{equation*}
we have
\begin{equation*}
	\input{img/6h_rewrite_crossing4.tikz}
\end{equation*}
to which rule $\mathrm{X}$ can be applied, yielding a diagram in normal form. The braidings are handled similarly.
\end{proof}

\begin{remark} If the initial diagram has no crossings, rule $7b$ is not needed for its normalisation.
\end{remark}

\begin{cor}[Completeness of the ZW calculus]
$F_\mathrm{ZW}: \mathbf{SD}[T_c/\mathrm{ZW}_c] \to \mathbf{Ab}_{2,\mathrm{free}}$ and $N: \mathbf{Ab}_{2,\mathrm{free}} \to \mathbf{SD}[T_c/\mathrm{ZW}_c]$ form a monoidal equivalence.
\end{cor}
\begin{proof}
We already know that $N$ is a right inverse for $F$. By uniqueness of the normal form, and soundness of the rewrite rules, if a diagram $\mathcal{G}$ is rewritten into $N(f)$ for some homomorphism $f$, then necessarily $f = F(\mathcal{G})$. 

Since, by Theorem \ref{thm:completeness}, such a rewrite is always possible, $\mathcal{G} = NF_\mathrm{ZW}(\mathcal{G})$ for all morphisms $\mathcal{G}$ of $\mathbf{SD}[T_c / \mathrm{ZW}_c]$.
\end{proof}

Although we only explicitly stated the completeness of the condensed ZW calculus, that of the expanded version immediately ensues.

One consequence that we can draw at once is that, under a suitable reinterpretation of the latter's diagrams, the ZW calculus contains the ZX calculus with $\pi$ phases, and is, to all effects, a refinement of it. This follows from the fact that a triangle of ternary W vertices corresponds to the ternary X vertex from the ZX calculus, with a $\pi$ phase:
\begin{equation*}
	\input{img/3d_triangle.tikz}
\end{equation*}

In particular, it is provable in the ZW calculus that 
\begin{equation*}
	\input{img/X_monoid.tikz}\quad \text{and} \quad \input{img/white_monoid.tikz}
\end{equation*}
form a \emph{strongly complementary pair} in the sense of \cite{coecke2012strong}. 

Moreover, the ZW calculus completes the axiomatisation of the GHZ/W calculus with additive inverses, as started in \cite{coecke2011ghz}, and can be used to encode rational arithmetic as suggested there.

With little effort, we can obtain completeness results for mild extensions of the ZW calculus. For all $n \in \mathbb{N}$, let $\mathbf{Ab}_{2,n}$ be the subcategory of $\mathbf{Ab}$ generated, under tensoring, by $\mathbb{Z}_n\oplus\mathbb{Z}_n$, where $\mathbb{Z}_n$ is the cyclic group of order $n$; and let $\mathrm{ZW}_n$ be the (expanded or condensed) ZW calculus augmented with the rule
\begin{equation*}
	\input{img/6i_loop_n.tikz}
\end{equation*}
There is a quotient functor $\mathbf{Ab}_{2,\mathrm{free}} \to \mathbf{Ab}_{2,n}$ induced by the quotient $\mathbb{Z} \twoheadrightarrow \mathbb{Z}_n$, and we can see that the rewrite rule $\mathrm{or}^n$ is precisely the implementation, on diagrams in normal form, of the action of this functor. Thus, we can state the following.

\begin{cor}
For all $n \in \mathbb{N}$, $\mathbf{SD}[T_c / \mathrm{ZW}_n]$ is monoidally equivalent to $\mathbf{Ab}_{2,n}$.
\end{cor}

The case $n = 2$ is particularly interesting, for it becomes provable that
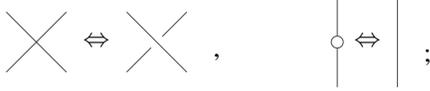
\begin{equation*}
	\input{img/6i_zw2.tikz}
\end{equation*}
these two rules, alone, can then replace $2b$, $3b$, $4$, $7a$, $7b$ and $\mathrm{X}$, leading to a significantly simplified calculus. In fact, even $\mathrm{or}^2$ becomes just a consequence of rule $5a$:
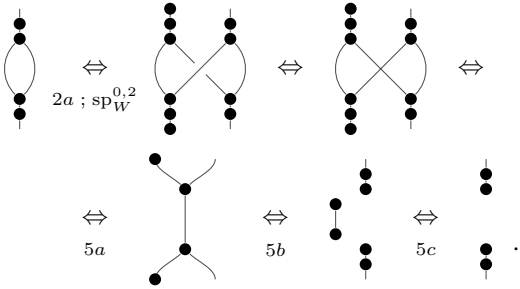
\begin{equation*}
	\input{img/6i_rule5a.tikz}
\end{equation*}
The category $\mathbf{Ab}_{2,2}$ was considered in \cite{schumacher2012modal} as a toy model of quantum theory - the theory of (pure) \emph{mobits}. The $\mathrm{ZW}_2$ calculus is a complete axiomatisation of it.

It is conceivable that the ZW calculus might be adapted to describe modules over more general rings and semirings. In the important example of $\mathbf{FRel}$, that is, modules over the semiring of Booleans, one obvious step would be to replace rule $5d$ with
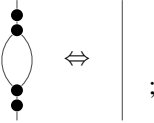
\begin{equation*}
	\input{img/6j_frel.tikz}
\end{equation*}
however, there is no such clear substitute for rule $5a$.

\section{Conclusions and outlook}

In this paper, we set out to improve and complete the axiomatisation started in \cite{coecke2011ghz, kissinger2012pictures} of the relations between the GHZ and W 3-qubit quantum states.

This led us to a new diagrammatic calculus, the ZW calculus, of which we defined two equivalent versions: expanded and condensed. We proved the soundness and completeness of the ZW calculus with respect to an interpretation in $\mathbf{Ab}_{2,\mathrm{free}}$, the category of abelian groups and homomorphisms generated by $\mathbb{Z} \oplus \mathbb{Z}$ under tensoring, by describing a normal form for its condensed version, and an explicit normalisation procedure. With that, we also proved that the ZW calculus refines a version of the ZX calculus, while retaining its symmetry and simple algebraic characterisation.

While this result may have a certain conceptual interest by itself, it is but one small step in a wider programme, which can be carried on in several directions.

Our normalisation procedure was tailored to making the completeness proof short and perspicuous, but it is by no means an efficient one. One obvious next step would be to study and improve the computational properties of the ZW calculus, looking for clever rewrite strategies, equivalent rulesets with a better performance, and, possibly, different normal forms. 

Indeed, our normal form, devised for the sake of the completeness proof, is exactly as informative as the sum expression (\ref{eq:sum}), and has none of the advantages of the diagrammatic notation for states, such as representing their \emph{separability} as topological disconnectedness, so it might be worth exploring some alternatives. This may be done with the help of Quantomatic \cite{kissinger2014quantomatic}.

On a different subject, to make the ZW calculus more useful for calculations, one would need a way to boost it from the integers to real numbers (or approximations thereof), and interpret it in the category of real vector spaces. The fact that
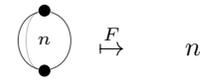
\begin{equation*}
	\input{img/7a_numbers.tikz}
\end{equation*}
suggests that wires are already used for counting in the ZW calculus, in the only way they possibly can, being measureless: one each. Then, a possibility that comes to mind is adding wires with a \emph{signed measure} on top, and defining
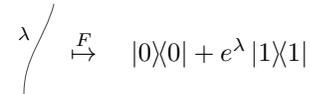
\begin{equation*}
	\input{img/7b_measure.tikz}
\end{equation*}
for a wire that is ``long'' $\lambda \in (-\infty, +\infty)$; this is similar to the ZX calculus with arbitrary phases. From here, and in the direction of SLOCC classification, the next step would be reaching complex numbers. One could just proceed in a similar fashion, adding phases like in the ZX calculus; but, possibly, a further ``geometrisation'' of the ZW calculus will suggest unexpected, more natural ways of encoding complex phases.

However, as much as it is worth investigating extensions of the ZW calculus, the same is true of its fragments. One that we mentioned before is the monochromatic fragment, consisting of black vertices and crossings, and whose interpretation is restricted to purely even and purely odd maps. Eliminating the second colour, in a way, leaves us with pure topology, and the fact that a self-crossing wire corresponds to the ``sign changer'' hints at some specific topological phenomenon lurking behind.

Moreover, there are hints that this topology might already contain indications for SLOCC classification. For tripartite states, 
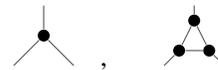
\begin{equation*}
	\input{img/7c_tripartite.tikz}
\end{equation*}
correspond to the two distinct maximally entangled SLOCC classes, and they very obviously have a different topology. Similarly, for quadripartite states, 
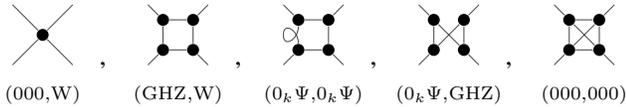
\begin{equation*}
	\input{img/7d_fourpartite.tikz}
\end{equation*}
are all representatives of distinct SLOCC super-classes, as defined in \cite{lamata2007inductive}, with the corresponding right singular subspace written below each diagram. We do not know, for now, exactly which states, in how many SLOCC classes, are expressible in the monochromatic language, but it might be worth tackling their classification first.

To end on a speculative note: the completeness of the expanded ZW calculus shows that $\mathbf{Ab}_{2,\mathrm{free}}$ is fully captured by undirected string diagrams with vertices of two colours - in fact, binary and ternary vertices suffice - and a few algebraically motivated axioms; and the ordering and directionality that are imposed by the categorical description come to be seen as redundant structure, over a simpler \emph{geometry of morphisms}. 

Taking this one step further, we wonder: is there an underlying geometry of the GHZ and W states that fully captures our axioms, in the way that the simpler theory of commutative Frobenius algebras is captured by 2-dimensional topological quantum field theories \cite{abrams1996two}? 

Understanding the compositional structure of multipartite entanglement is likely to involve an original interplay of algebra and geometry; monoidal categories, with their associated diagrammatic languages, might just provide the bridge that is needed.

\section*{Acknowledgment}

The author is supported by an EPSRC Doctoral Training Grant. Many thanks to Bob Coecke, Stefano Gogioso, Aleks Kissinger, Jamie Vicary, and Linde Wester for useful discussions and suggestions. All diagrams were drawn with TikZiT \cite{kissinger2014tikzit}, whose developers also have the author's gratitude.

\bibliographystyle{IEEEtran}
\bibliography{IEEEabrv,bibliography}

\end{document}